\documentclass[a4paper,UKenglish]{article}
\usepackage{amsmath,amsthm,amsfonts,graphicx,hyperref,color,caption,subcaption,float,algpseudocode,soul,multirow, fullpage}
\usepackage{authblk}

\usepackage{microtype,amssymb}
\usepackage{framed}
\newtheorem{theorem}{Theorem}
\newtheorem{lemma}[theorem]{Lemma}

\newtheorem{corollary}[theorem]{Corollary}
\newtheorem{definition}[theorem]{Definition}

\newtheorem*{remark}{Remark}
\def\RR{{\mathbb R}}

\def\EE{{\mathbb E}}
\def\SS{{\mathbb S}}
\def\Pr{{\mathrm Pr}}
\def \d{{\mathrm d}}
\def \exp{{\mathrm exp}}
\def \e{{\mathrm e}}
\sethlcolor{white}
\title{Products of Euclidean metrics and applications to proximity questions among curves}
\author{Ioannis Z.~Emiris\thanks{Department of Informatics \& Telecommunications, National \& Kapodistrian University of Athens, and ATHENA Research Center, Maroussi, Greece. } \qquad Ioannis Psarros\thanks{Department of Informatics \& Telecommunications, National \& Kapodistrian University of Athens.}}

\begin{document}
\maketitle

\begin{abstract}
The problem of Approximate Nearest Neighbor (ANN) search is fundamental in computer science and has benefited from significant progress in the past couple of decades. However, most work has been devoted to pointsets whereas complex shapes have not been sufficiently treated.
Here, we focus on distance functions between discretized curves in Euclidean space: they appear in a wide range of applications, from road segments to time-series in general dimension.
For $\ell_p$-products of Euclidean metrics, for any $p$, we design simple and efficient data structures for ANN, based on randomized projections, which are of independent interest.
They serve to solve proximity problems under a notion of distance between discretized curves, which generalizes both discrete Fr\'{e}chet and Dynamic Time Warping distances.
These are the most popular and practical approaches to comparing such curves.
We offer the first data structures and query algorithms for ANN with arbitrarily good approximation factor, at the expense of increasing space usage and preprocessing time over existing methods. Query time complexity is comparable or significantly improved by our algorithms; our algorithm is especially efficient when the length of the curves is bounded.

\bigskip


\end{abstract}

\section{Introduction}

The problem of Nearest Neighbor (NN) search is fundamental in computer science: one has to preprocess a dataset so as to answer proximity queries efficiently, for any given query object. The Approximate Nearest Neighbor (ANN) problem is a relaxation of the above problem where proximity queries are answered approximately:
a valid answer realizes a distance to the query which can be  larger than the distance between the query and its nearest neighbor, but in a controllable way.  
ANN has been enjoying a lot of attention, and significant progress has been achieved in the past couple of decades, both in the theoretical as well as in the practical realm. However, most work has been devoted to vector spaces, and complex objects have not been sufficiently treated in the literature.
Here, in order to address one such class of complex objects, we focus on distance functions for polygonal curves, which lie in the Euclidean space. Polygonal curves are essentially point sequences of varying length, and have a wide spectrum of applications, ranging from road segments in low dimensions (e.g.~\cite{giscup2017}) to time-series in arbitrary dimension (e.g.~\cite{HK02}) and up to protein backbone structures (e.g.~\cite{JXZ08}).
In general, the problem we aim to solve is as follows.

\begin{definition}[ANN]\label{Dgenann}
The input consists of $n$ polygonal curves $V_1,\ldots,V_n$, where each $V_i$ is a sequence $v_{i1},\ldots,v_{im_i}$ with each $v_{ij}\in\RR^d$, and each $m_i\leq m$ for some pre-specified $m$. Given distance function $\d(\cdot,\cdot)$, and approximation parameter  $\epsilon>0$, preprocess $V_1,\ldots,V_n$ into a data structure such that for any query polygonal curve $Q$, the data structure reports $V_j$ for which the following holds
\[
\forall i\in\{1,\dots,n\}:\; ~ \d(Q,V_j)\leq(1+\epsilon)\cdot \d(Q,V_i). 
\]
\end{definition}

 There are various ways to define dissimilarity or distance between two curves. Two very popular dissimilarity measures are the Discrete Fr\'{e}chet Distance (DFD) and the Dynamic Time Warping (DTW) distance, which are both widely studied and applied to classification and retrieval problems for various types of data. DFD is a distance metric, unlike DTW which does not satisfy the triangular inequality.
 
 It is common, in distance functions of curves, to involve the notion of a traversal for two curves. Intuitively, a traversal corresponds to a time plan for traversing the two curves simultaneously, starting from the first point of each curve and finishing at the last point of each curve. With time advancing, the traversal advances in at least one of the two curves. 
 Every traversal consists of a sequence of pairs of points, lying in each of the two curves, that have been visited at the same time while traversing.
 DFD minimizes, over all traversals, the maximum Euclidean distance between two points defining a pair in a traversal. DTW minimizes, over all traversals, the sum of Euclidean distances between points defining a pair in a traversal. 

We denote by $\ell_p^d$ the normed space $(\RR^d, \| \cdot\|_p)$, where for any $x=(x_1,\ldots,x_d)\in \RR^d$, $\|x\|_p=\left(\sum_i |x_i|^p\right)^{1/p}$. We also use the notation $\tilde{O}(f(d,m,n))$, which hides polylogarithmic factors in $f(d,m,n)$ and polynomial factors in $1/\epsilon$. 

\subsection{Previous work}
The ANN problem has been mainly addressed for datasets consisting of points. Efficient deterministic solutions exist when the dimension is bounded, e.g.~\cite{AMM09,AdFM11,HIM12} and are based on a notion of Approximate Voronoi Diagrams (AVDs) or tree-based data structures, while for high-dimensional data the state-of-the-art solutions are mainly based either on the concept of Locality Sensitive Hashing (LSH), e.g.\ \cite{HIM12,ALRW17}, or on random projections, e.g.\ \cite{IN07, AC09, AEP18} extending the celebrated result by Johnson and Lindenstrauss~\cite{JL84}. Another line of work focuses on subsets of general metrics which satisfy some sort of low intrinsic dimension assumption, e.g.~\cite{HPM05}.
The \emph{doubling dimension} of a metric space is defined as the logarithm of the minimum number of balls of radius $R/2$ which can cover any ball of radius $R$ in this metric space. We formally define and use this notion in Section~\ref{ssection:doub}, since it provided the main paradigm today for exploiting low intrinsic dimension.

We have listed only a fraction of the body of work that is available today on pointsets; however, much less is known about distances between curves which, in a sense, are the next more complex type of geometric object.

Let us start with point sequences, which are closely related to curves. For metrics $(M_1,\d_{1}),\ldots,(M_k,\d_{k})$, we define the $\ell_p$-product of $(M_1,\d_{1}),\ldots,(M_k,\d_{k})$ as the metric with domain $M_1\times\cdots\times M_k$ and distance function 
\[
\d((x_1,\ldots,x_k),(y_1,\ldots,y_k))=\left( \sum_{i=1}^k \d_{i}^p(x_i,y_i)\right)^{1/p} .
\]
If there exists an ANN data structure with approximation factor $c$ for each $(M_1,\d_{1}),\ldots,(M_k,\d_{k})$, then one can build a data structure for the $\ell_p$-product with approximation factor $O(c\log\log n)$ \cite{Athes09,I02}. 

Let us now focus on curves: several data structures aiming to solve the 
exact nearest neighbor problem for the (Discrete) Fr\'echet Distance have been devised, but they are either heuristics or they come with weak performance guarantees   \cite{giscup2017,BaldusB2017,DutschV17,BuchinDDM17,BGM17}. 
The related range searching problem in two dimensions has been  recently studied \cite{AD18}. Notice that when the assumed distance function defines a metric over the domain of curves, solutions for general metrics can be applied (see e.g. \cite{KR02, BKL06}), but further analysis on the performance would be required. In any case, the distance functions studied in this paper do not always define a metric: in particular, DTW does not satisfy the triangular inequality.   

 All existing {\em approximate} solutions solve the approximate {\em near} neighbor problem, which is essentially a decision problem, instead of the optimization ANN problem. In the approximate near neighbor problem, the input consists of a set of $n$ polygonal curves $\mathcal{P}$, an approximation parameter $\epsilon>0$, and a radius parameter $r>0$. A data structure for this problem is required to support the following type of query: for any query curve $Q$, if there exists a curve in $\mathcal{P}$ within distance $r$, then report any curve within distance $(1+\epsilon)r$, and if there is no curve within distance $(1+\epsilon)r$, then report "no".  
 It is known that a data structure for the approximate near neighbor problem can be used as a building block for solving the ANN problem. This procedure has provable guarantees on metrics~\cite{HIM12}, but it is not clear whether it can be extended to non-metric distances such as the DTW. 

The first result for DFD by Indyk~\cite{I02}, for any metric $(X,d(\cdot,\cdot))$, achieved an approximation factor of
$O((\log m + \log \log n)^{t-1})$, where $m$ is the maximum length of a curve, and $t>1$ is a trade-off parameter. 
The solution is based on an efficient data structure for $\ell_{\infty}$-products of arbitrary metrics, and achieves space and preprocessing time in $O(m^2 |X|)^{tm^{1/t}}\cdot n^{2t}$, and query time in $(m \log n)^{O(t)}$.
Table~\ref{tab:compar} states these bounds for appropriate $t=1+o(1)$, hence for a constant approximation factor.
 It is not clear whether the approach may achieve a $1+\epsilon$ approximation factor by employing more space.
 
More recently, a data structure was devised for the DFD of curves defined by the Euclidean metric~\cite{DS17}. The approximation factor is $O(d^{3/2})$. The space required is $O(2^{4md} n \log n+mn)$ and each query costs $O(2^{4md} m \log n)$. They also provide a trade-off between space/query time, and the approximation factor. For space in $O( n \log n+mn)$, they achieve query time in $O(m \log n)$ and approximation factor in  $O(m)$. Our methods can achieve any user-desired approximation factor at the expense of a reasonable increase in the space and time complexities.

Furthermore, it is shown that the result establishing an $O(m)$ approximation \cite{DS17} extends to DTW, whereas the other extreme of the trade-off has remained open.
 
Table~\ref{tab:compar} summarizes space and query time complexities, and approximation factors of the main methods for searching among discrete curves under the two main dissimilarity measures.

\subsection{Our contribution}

Our first contribution is a simple data structure for the ANN problem in $\ell_p$-products of finite subsets of $\ell_2^d$, for any constant $p$. The key ingredient is a random projection from points in $\ell_2$ to points in $\ell_p$. Although this has proven a relevant approach for ANN of pointsets, it is quite unusual to employ randomized embeddings from $\ell_2$ to $\ell_p$, $p>2$, in the context of proximity searching because no better data structures are known for such norms. In particular, the extreme case $p=\infty$ is provably harder than the cases $p\in[1,2]$ \cite{ACP08}.  
After the random projection, the algorithm ``vectorizes" all point sequences. The original problem is then translated to the ANN problem for points in $\ell_p^{d'}$, for $d'\approx d\cdot m$ to be specified later, and can be solved by simple bucketing methods in space $\tilde{O}\left(d' n \cdot (1/\epsilon)^{d'} \right)$ and query time $\tilde{O}(d' \log n)$, which is very efficient when the product $ d\cdot m$ is bounded. 

Then, we present a notion of distance between two polygonal curves, which generalizes both DFD and DTW (for a formal definition see Definition \ref{Ddist}): The \textit{$\ell_p$-distance} of two curves minimizes, over all traversals, the $\ell_p$ norm of the vector of all Euclidean distances between paired points. Hence, DFD corresponds to $\ell_{\infty}$-distance of polygonal curves, and DTW corresponds to $\ell_1$-distance of polygonal curves. 

Our main contribution is an ANN data-structure for the $\ell_p$-distance of curves, when $1\leq p<\infty$. This easily extends to $\ell_\infty$-distance of curves by solving for the $\ell_{p}$-distance, where $p$ is sufficiently large. 
Our target are methods with approximation factor $1+\epsilon$. Such approximation factors are obtained for the first time, at the expense of larger space or time complexity than in previous methods.
Moreover, a further advantage is that our methods solve ANN directly instead of requiring to reduce it to near neighbor search. We recall that a reduction to the near neighbor problem has provable guarantees on metrics~\cite{HIM12}, however we are not aware of an analogous result for non-metric distances such as the DTW.

Specifically, when $p>2$, we show that there exists a data structure with space and preprocessing time in 
\[
\tilde{O}\left(n\cdot \left( \frac{d}{p\epsilon} + 2\right)^{O\left(d m \cdot \alpha_{p,\epsilon}\right)}\right),
\]
where $\alpha_{p,\epsilon}$ depends only on $p,\epsilon$, and {query time} in $\tilde{O}(d\cdot 2^{4m + p}\log n)$. 

When specialized to DFD and juxtaposed to~\cite{DS17}, the two methods are only comparable when $\epsilon$ is a large enough fixed constant. Indeed, 
the two space and preprocessing time complexity bounds are equivalent, i.e.~they are both exponential in $d$ and $m$, but our query time is linear instead of being exponential in $d$.

When $p\in[1,2]$, there exists a data structure with space and preprocessing time in
\[\tilde{O}\left(n\cdot 2^{O\left(d m \cdot \alpha_{p,\epsilon}\right)}\right),\]
where $\alpha_{p,\epsilon}$ depends only on $p,\epsilon$, and {query time} in $\tilde{O}\left(d\cdot 2^{4m}\log n\right)$. 
This leads to the first approach that achieves $1+\epsilon$ approximation for DTW at the expense of space, preprocessing and query time complexities being exponential in $m$. Hence our method is best suited when the curve size is small.

Lastly, we focus on DFD when the ambient space has high dimension, and we discuss complexity bounds in terms of the doubling dimension, while we also design an improved algorithm for approximate near neighbor search. 

\hspace*{-1cm}\begin{table}\begin{center} 
\begin{tabular}{|l|l|l|c|l|} \hline
  & Space & Query & Approx. & Comments \\ \hline\hline
\multirow{ 5}{*}{DFD}  	&  ${O}( m^2 |X| )^{m^{1-o(1)}}\cdot O( n^{2-o(1)})$    &$\left(m \log n\right)^{O(1)}$ &  $O(1)$ & any metric, determ.~\cite{I02}
\\ 
	& $\tilde{O}(2^{4md} n )$ &$\tilde{O}(2^{4md} \log n )$  &  $O(d^{3/2})$  & $\ell_2^d$, \cite{DS17}
    \\ 
 &  $\tilde{O}( d  m^2 n ) \left( \frac{d}{\log m} +2\right)^{O(m^{O(1/\epsilon)} d \log(1/\epsilon))}$ &  $\tilde{O} (dm^{O(1/\epsilon)} 2^{4m} \log n ) $ & $1+\epsilon$ & $\ell_2^d$,  Thm~\ref{Tdfdann} 
    \\ 
 &  $\tilde{O}\left( d m n \right) \left( \frac{ \log (2/\epsilon)}{\epsilon^2} +2\right)^{O(m^{O(1/\epsilon)}\log^2(2/\epsilon) )}$ &  $\tilde{O}\left(d m^{O(1/\epsilon)} 2^{4m}\log n\right) $ & $1+\epsilon$ & $\ell_2^d$, $ddim=O(1)$, Thm~\ref{TannDFDdd} 
    \\ 
 &  $(nm)^{O(m\epsilon^{-2})} +O(dmn)$ &  $\tilde{O}\left(d\cdot 2^{4m}\cdot \log n\right)$ & $1+\epsilon$ & $\ell_2^d$,  Thm~\ref{TannDFDhd} 
 
   \\  \hline
\multirow{3}{*}{DTW}  &  ${\tilde{O}}( dmn )$  &  ${O}(m \log n)$  &   $O(m)$ & $\ell_2^d$, ~\cite{DS17} 
\\ 
&  $\tilde{O}(d m^2 n )\cdot 2^{O(m \cdot d  \log(1/\epsilon))}$  &  $\tilde{O}(d  \cdot 2^{4m} \log n)$ &   $1+\epsilon$  & $\ell_2^d$,  Thm~\ref{Tdtwann} 
\\ 
& $\tilde{O}(d\cdot 2^{4m}n^{1+\rho_u})$ & $\tilde{O}\left(d\cdot 2^{4m} n^{\rho_q}\right)$ & $1+\epsilon$ & $\ell_2^d$,  Thm~\ref{Tdtwann2} 
\\\hline
\end{tabular}
\caption{\label{tab:compar}Summary of previous results compared to this paper's: $X$ denotes the domain set of the input metric. 
All methods are randomized except for the first one (Indyk's).
We denote by $ddim$ the doubling dimension of the input space. 
All previous results have been tuned to optimize the approximation factor. Parameters $\rho_u$, $\rho_q$ are chosen to satisfy $(1+\epsilon) \sqrt{\rho_q}+\epsilon\sqrt{\rho_u}\geq \sqrt{1+2\epsilon}$.}
\end{center}
\end{table}

Our results for DTW and DFD are summarized in Table~\ref{tab:compar} and juxtaposed to the existing approaches, which were proposed in~\cite{DS17,I02}.

This paper is a follow-up of \cite{EmirisP18}. It consists of refined and polished versions of the results and the proofs that appear there. Some issues are fixed, e.g.~the dependence on $m$ in the complexity bounds for DFD, while new results have been added, namely the entire Section~\ref{Sechd}. 

The rest of the paper is structured as follows.
In Section \ref{Sseqs}, we present a data structure for ANN in $\ell_p$-products of $\ell_2$, which is of independent interest. In Section~\ref{Scurves}, we employ this result to address the $\ell_p$-distance of curves. In Section~\ref{Sechd}, we focus on DFD when the ambient space is high dimensional.
We conclude with future work.

\section{$\ell_p$-products of $\ell_2$}\label{Sseqs}

In this section, we present a simple data structure for the ANN problem in $\ell_p$-products of finite subsets of $\ell_2$. Recall that the $\ell_p$-product of $X_1,\ldots,X_m$, which are finite subsets of $\ell_2$, is a metric space with ground set $X_1 \times X_2 \times \cdots \times X_m$ and distance function: 
\[
\!\!\!\! \d((x_1,\ldots,x_m),(y_1,\ldots,y_m))=\left\| \|x_1-y_1\|_2, \|x_2-y_2\|_2, \ldots, \|x_m-y_m\|_2  \right\|_p = \left(\sum_{i=1}^m \|x_i-y_i\|_2^p\right)^{1/p}.
\] 
The algorithm first randomly embeds points from $\ell_2$ to $\ell_p$. 
Then, it is easy to reduce the original problem to an ANN problem in $\ell_p$ for large vectors corresponding to point sequences. 

For a random variable $X$ and some probability distribution $\mathcal{D}$, we use $X\sim \mathcal{D}$ to denote the fact that $X$ follows $\mathcal{D}$, while $N(\mu,\sigma^2)$ denotes the normal distribution with mean $\mu$ and variance $\sigma^2$. 
Moreover, $\EE[\cdot]$ denotes expectation, and $\langle \cdot,\cdot\rangle$ denotes inner product.

\subsection{Concentration Inequalities}\label{SSconc}

In this subsection, we prove concentration inequalities for central absolute moments of the normal distribution. Some of these results may be folklore, and the reasoning is quite similar to the one followed by proofs of the Johnson-Lindenstrauss lemma, e.g.~\cite{Mat08}. Notice also that results concerning random projections from $\ell_2$ to $\ell_p,p\in[1,2]$ are folklore, but we are also interested in the case $p>2$. In addition, the properties which are required for ANN searching are weaker than the ones which are typically investigated. 

The $2$-stability property of standard normal variables, along with standard facts about their absolute moments, imply the following claim. 
\begin{lemma}\label{ClmExp}
Let $v \in \RR^d$ and let $G$ be $k \times d$ matrix with i.i.d random variables following $N(0,1)$. Then,
\[
\EE\left[\|Gv\|_p^p\right]= c_p\cdot k \cdot  \|v\|_2 ^p,
\]
where $c_p=\frac{2^{p/2} \cdot \Gamma\left(\frac{p+1}{2}\right)}{\sqrt{\pi}}$ is a constant depending only on $p>1$. 
\end{lemma}
\begin{proof}
Let $g=(X_1,\ldots,X_d)$ be a vector of independent random variables which follow $N(0,1)$ and let $v$ be any vector in $ \RR^d$. The $2$-stability property of Gaussian random variables implies that $\langle g,v\rangle\sim N(0,\|v \|_2^2 )$. Recall the following standard fact for central absolute moments of $Z\sim N(0,\sigma^2)$:
\[
\EE\left[|Z|^p\right]= \sigma^p \cdot \frac{2^{p/2} \cdot \Gamma\left(\frac{p+1}{2}\right)}{\sqrt{\pi}}.
\]
Hence,
 \[
\EE\left[\|Gv\|_p^p\right]= \EE \left[\sum_{i=1}^k |\langle g_{i},v \rangle|^p\right] = k \cdot \|v\|_2^p \cdot \frac{2^{p/2} \cdot \Gamma\left(\frac{p+1}{2}\right)}{\sqrt{\pi}}.
\]
\end{proof}

In the next lemma, we offer a simple upper bound on the moment generating function of $|X|^p$, where $X\sim N(0,1)$.  

\begin{lemma}
\label{lemma:gaussianmgf}
Let $X\sim N(0,\sigma^2)$, $p\geq1$, and $t>0$, then $\EE[ exp(-t|X|^p)]\leq 
exp(-t\EE[|X|^p]+t^2\EE[|X|^{2p}])$.
\end{lemma}

\begin{proof}
We use the easily verified fact that for any $x\leq 1$, $exp(x)\leq 1+x+x^2$ and the standard inequality $1+x\leq \e^{x}$, for all $x\in \RR$. 
Then,
\[\EE\left[ \mathrm{e}^{-t|X|^p}\right]
\leq 1-t\cdot \EE\left[|X|^p\right]+t^2 \cdot \EE\left[|X|^{2p}\right] \leq \e^{-t\EE[|X|^p]+t^2\EE[|X|^{2p}]}.
\]
\end{proof}

In the next lemma, we bound the expectation of a certain variable to a squared power, by a function of the squared expectation of this variable's power. In particular we bound the ratio 
 $\EE[|X|^{2p}]\, /\, \EE[|X|^{p}]^2 =$  $O \left(
\Gamma\left(\frac{2p+1}{2}\right) \cdot \Gamma\left(\frac{p+1}{2}\right)^{-2} \right)$ by making use of Stirling estimates.

\begin{lemma}\label{ClmMom}
Let $X\sim N(0,1)$. Then, there exists constant $C>0$ such that for any integer $p\geq1$, $\EE[|X|^{2p}]\leq C \cdot 2^p\cdot \EE[|X|^{p}]^2$. In addition, 
there exists constant $C>0$ such that for any real $p\geq1$, $\EE[|X|^{2p}]\leq C \cdot p^2\cdot 2^p\cdot \EE[|X|^{p}]^2$.
\end{lemma}
\begin{proof}
In the following, we denote by $f(p) \approx g(p)$ the fact that $f(p)=\Theta(g(p))$. In addition, $f(p)\gtrsim g(p)$ means that $f(p) = f(p) = \Omega(g(p))$ and $g(p)=O(f(p))$. 
Notice that if $p$ is a bounded constant then the statement trivially holds, so we focus on bounding the ratio for sufficiently large $p \gg 1$.  

First we consider the case of integer $p$. 
In what follows, we use the Stirling estimate $n! \approx n^{n+1/2} \e^{-n}$ and standard facts about moments of normal variables:
\[
\EE\left[|X|^{2p}\right]=\frac{2^{p} \cdot \Gamma\left(\frac{2p+1}{2}\right)}{\sqrt{\pi}} \approx (2p-1)!!=\frac{(2p)!}{2^p \cdot p!}
\approx\left( \frac{2p}{\e}\right)^{2p} \cdot \frac{1}{2^p \cdot \left( \frac{p}{\e}\right)^p } \approx \frac{2^p  p^p }{\e^p} ~.
\]
We now analyze
\[
\EE\left[|X|^{p}\right]^2=
\left(
\frac{2^{p/2} \cdot \Gamma\left(\frac{p+1}{2}\right)}{\sqrt{\pi}} \right)^2
\approx\left( (p-1)!!\right)^2,
\]
by considering two cases: either $p$ is odd or $p$ is even.  
If $p$ is odd, then $p-1$ is even, hence 
\[
\left( (p-1)!!\right)^2 \approx \left( 2^{\frac{p-1}{2}} \cdot \left( \frac{p-1}{2}\right)! \right)^2 \approx   
\left( 2^{\frac{p-1}{2}} \cdot \left( \frac{p-1}{2}\right)^{p/2} \cdot \e^{-\frac{p-1}{2}} \right)^2
\approx
\frac{(p-1)^{p}}{\e^{p}} ~.
\]
If $p$ is even, then 
\[
((p-1)!!)^2 = \left(\frac{p!}{2^{p/2} \cdot \left(\frac{p}{2} \right)! } \right)^2 \approx \left( \frac{p^{p+1/2} \cdot  \e^{-p}}{2^{p/2} \cdot \left( \frac{p}{2}\right)^{p/2+1/2} \cdot \e^{-p/2} }\right)^2 \approx \frac{p^{p}}{\e^{p}}. 
\]
Hence, putting everything together,
\[
\frac{\EE\left[|X|^{2p}\right]}{\EE\left[|X|^{p}\right]^2} \lesssim 2^p \cdot \left(\frac{p}{p-1}\right)^p \lesssim 2^p,
\]
for sufficiently large $p$. 

If $p$ is any real number, we have:
\[
\frac{\EE\left[|X|^{2p}\right]}{\EE\left[|X|^{p}\right]^2} \leq 
\frac{\EE\left[|X|^{2\lceil p \rceil}\right]}{\EE\left[|X|^{\lfloor p \rfloor}\right]^2} 
\lesssim 2^p \cdot \left(\frac{(p+1)^{p+1}}{(p-2)^{p-1}}\right) \lesssim 2^p \cdot p^2 .
\]
\end{proof}

\begin{remark}
Lemma~\ref{ClmMom} provides two upper bounds on the ratio $\EE\left[|X|^{2p}\right]/\EE\left[|X|^p \right]^2$: one holds for any real  $p\geq 1$ whereas the other is slightly improved for any integer $p \geq 1$. 
We remark that throughout this paper, we focus on the general case that $p\geq 1$ is real, but slightly better bounds can be obtained if $p$ is an integer.  
\end{remark}

Now we pass to a tail inequality, a common tool in establishing the properties of embeddings.
The following lemma is the main ingredient of our analysis, since it provides us with a lower-tail inequality for one projected vector. The key observation here is that the random variable $\|Gv\|_p^p$ is essentially a sum of independent Gaussian central  absolute moments. 

\begin{lemma}\label{LemConP}
Let $G$ be a $k\times d$ matrix with i.i.d.~random variables following $ N(0,1)$ and consider vector $v \in \RR^d$, such that $\|v\|_2=1$. {For appropriate constant $c'>0$}, for $p\geq 1$ and $\delta \in (0,1)$, it holds
\[
\Pr\left[\|Gv\|_p^p\leq (1-\delta) \cdot \EE \left[ \|Gv\|_p^p \right] \right] \leq \e^{-c' \cdot 2^{-p}\cdot p^{-2}\cdot k\cdot \delta^2}. 
\]
\end{lemma}
\begin{proof}
As discussed above, in studying the behavior of matrix-vector product $Gv$ it suffices to study some normal variable $X$. Formally, for $X\sim N(0,1)$ and any $t>0$, 
\[
\Pr\left[\|Gv\|_p^p\leq (1-\delta) \cdot \EE \left[ \|Gv\|_p^p\right]\right] \leq \EE\left[ \e^{-t|X|^p}\right]^k \cdot \e^{(t(1-\delta) k \cdot \EE[|X|^p])}
\leq
\e^{k(-t \cdot \EE[|X|^p]+t^2 \cdot C\cdot 2^p \cdot p^2 \cdot \EE[|X|^p]^2 + t\cdot (1-\delta)\cdot \EE[|X|^p])} ,
\]
where the first inequality holds by Markov's inequality extended to the exponential function and the 2-stability property of standard Gaussian random variables, and 
the second inequality follows by employing Lemma~\ref{lemma:gaussianmgf} and then  Lemma~\ref{ClmMom}, where $C$ is the constant appearing in Lemma~\ref{ClmMom}. 
Now, we set $t=\frac{\delta}{2 \cdot {C} \cdot 2^p \cdot p^2 \cdot \EE[|X|^p]}$, and the exponent simplifies to one term proportional to $-\delta^2$, which dominates terms linear in $\delta$. Hence, 
\[
\Pr\left[\|Gv\|_p^p\leq (1-\delta) \cdot \EE\left[ \|Gv\|_p^p\right]\right] \leq \e^{-c'\cdot 2^{-p} \cdot p^{-2} \cdot k\cdot \delta^2}, 
\]
for some constant $c'>0$. 
\end{proof}

Finally, we make use of the following one-sided Johnson-Lindenstrauss lemma (see, e.g.,~\cite{Mat08}).

\begin{theorem} \label{ThmJL}
Let $G$ be a $k\times d$ matrix with i.i.d.~random variables following $N(0,1)$ and consider vector $v \in \RR^d$. 
Then, for constant $C>0$,
\[
\Pr\left[\|Gv\|_2 \geq (1+\epsilon)\|v\|_2  \sqrt{k} \right]\leq {\e^{-C\cdot k\cdot \epsilon^2}},
\]
where $\epsilon \in (0,1/2]$.
\end{theorem}

Standard properties of $\ell_p$ norms imply the loose upper tail inequality below.

\begin{corollary}\label{CorWeakUpT}
Let $p\geq 2$. 
Let $G$ be a $k\times d$ matrix with i.i.d.~random variables following $ N(0,1)$ and consider vector $v \in \RR^d$. Then, for constant $C>0$,
\[
\Pr\left[\|Gv\|_p\geq(1+\epsilon)\|v\|_2 \sqrt{k}\right]\leq {\e^{-C\cdot k\cdot \epsilon^2}},
\]
where $\epsilon \in (0,1/2]$.
\end{corollary}
\begin{proof}
Since $p\geq 2$, we have that $\forall x\in\RR^d$, $\|x\|_p\leq \|x\|_2$. Hence, by Theorem~\ref{ThmJL},
\[
\Pr\left[\|Gv\|_p\geq(1+\epsilon)\|v\|_2\sqrt{k}\right]\leq \Pr\left[\|Gv\|_2\geq(1+\epsilon)\|v\|_2\sqrt{k}\right]\leq {\e^{-C\cdot k\cdot \epsilon^2}}.
\]
\end{proof}
Furthermore, a slightly different loose upper tail inequality can be derived when $p\in [1,2]$.
\begin{lemma}\label{LemWeakUpT2}
Let $p\in [1,2]$. 
Let $G$ be a $k\times d$ matrix with i.i.d.~random variables following $ N(0,1)$ and consider vector $v \in \RR^d$. Then, for constant $C>0$,
\[
\Pr\left[\|Gv\|_p\geq (3\cdot c_p \cdot k)^{1/p}\|v\|_2 \right]\leq {\e^{-C\cdot k}},
\]
where $c_p= \frac{2^{p/2} \cdot \Gamma\left(\frac{p+1}{2}\right)}{\sqrt{\pi}}.$
\end{lemma}
\begin{proof}

Let $X\sim N(0,1)$. Then, using $p\le 2$ to get the inequality between integrals, we have:
\begin{equation} \label{equation:mgf1}
\EE\left[ \e^{|X|^p/3}\right]=\frac{1}{\sqrt{2 \pi}} \int_{-\infty}^{+\infty}\e^{|x|^p/3-x^2/2} \d x\leq\frac{\sqrt{2}}{\sqrt{\pi}} \int_0^{+\infty}\e^ {x^2/3-x^2/2} \d x =\frac{\sqrt{2}}{\sqrt{\pi}} \int_0^{+\infty}\e^ {-x^2/6} \d x= \frac{\sqrt{2}}{\sqrt{\pi}} \cdot \frac{\sqrt{3\pi}}{\sqrt{2}} =  \sqrt{3}.
\end{equation}
We now show a lower bound on $c_p=\frac{2^{p/2} \cdot \Gamma\left(\frac{p+1}{2}\right)}{\sqrt{\pi}}$ for any $p\in [1,2]$. 
We first focus on the behavior of the gamma function with input $z\in [1,3/2]$:
\[
\Gamma(z)=\int_{0}^{\infty} x^{z-1} \e^{-x} \d x=
\int_{0}^{1} x^{z-1} \e^{-x} \d x + \int_{1}^{\infty} x^{z-1} \e^{-x} \d x \geq 
\int_{0}^{1} x^{3/2} \e^{-x} \d x + \int_{1}^{\infty}  \e^{-x} \d x \geq 0.73 ~.
\]
Hence, for any $p\in[1,2]$, 
\begin{equation} \label{equation:mgf2}
    c_p=\frac{2^{p/2} \cdot \Gamma\left(\frac{p+1}{2}\right)}{\sqrt{\pi}} \geq \sqrt{2/\pi} \cdot 0.73 \geq 0.58~.
\end{equation}
Now, assume wlog that $\|v\|_2=1$. Then,
\[
\Pr\left[\|Gv\|_p^p\geq 3 \cdot \EE \left[ \|Gv\|_p^p\right]\right] \leq \EE\left[ \e^{|X|^p/3}\right]^k \cdot \e^{-  k \cdot \EE[|X|^p]}
\leq\frac{3^{k/2}}{\e^{k\cdot c_p }}
\leq \e^{-k( c_p-0.55)}\leq \e^{-k/100},
\]
where the first inequality holds by Markov's inequality extended to the (monotonically increasing) exponential function; the second inequality holds by using Inequality~(\ref{equation:mgf1}) for the numerator and Lemma~\ref{ClmExp} to establish the denominator; 
whereas in the last inequality, we used Inequality~(\ref{equation:mgf2}).
\end{proof}

\subsection{Embedding $\ell_2$ into $\ell_p$}\label{SSemb}

In this subsection, we present our main results concerning ANN for $\ell_p$-products of $\ell_2$. First, we show that a simple random projection maps points from $\ell_2^d$ to $\ell_p^{k}$, where $k = \tilde{O}(d)$, without contracting norms in an arbitrary fashion, most of the time. In particular, the probability of failure of this scheme decays exponentially with $k$. 
For our purposes, there is no need for an almost isometry between norms. Hence, our efforts focus on proving lower tail inequalities which imply that, with good probability, no point lying far away in the original space may correspond to an approximate nearest neighbor in the projected space. 

We now prove specific bounds concerning the contraction of distances of the embedded points. Our proof builds upon the inequalities developed in Subsection~\ref{SSconc}.

\begin{theorem}\label{ThmEmb}
Let $G$ be a $k\times d$ matrix with i.i.d.~random variables following $N(0,1)$. Then,
\begin{itemize}
\item if  $2<p<\infty$ then, 
\[\Pr \left[ \exists v \in \RR^d:~\|Gv\|_p \leq \frac{(c_p\cdot k)^{1/p} }{1+\epsilon} \cdot \|v\|_2 \right] \leq O\left( \frac{k^{\frac{1}{2}-\frac{1}{p}}}{p\epsilon }+2\right)^d \cdot {\e^{-c' \cdot 2^{-p} \ \cdot k\cdot(\epsilon/(2+p\epsilon))^2}},\]
\item if $p\in[1,2]$ then, 
\[\Pr \left[ \exists v \in \RR^d:~\|Gv\|_p \leq \frac{(c_p\cdot k)^{1/p} }{1+\epsilon} \cdot \|v\|_2 \right] \leq O\left( \frac{1}{\epsilon }\right)^d \cdot {\e^{-c' \cdot k\cdot(p\epsilon/(2+p\epsilon))^2}},\]
\end{itemize}

where 
$c'>1$ is a constant, $\epsilon\in(0,1/2]$.
\end{theorem}
\begin{proof} 
By Lemma~\ref{ClmExp}, we have that $\EE\left[\|Gv\|_p^p\right]= c_p\cdot k \cdot  \|v\|_2 ^p$, 
where $c_p=\pi^{-1/2}{2^{p/2} \cdot \Gamma\left(\frac{p+1}{2}\right)}$.  
We now employ Lemma~\ref{LemConP} to derive the second inequality below, for $\delta=\frac{\epsilon p}{2+p\epsilon}$:
\[
 \Pr \left[ \|Gv\|_p^p \leq \frac{c_p\cdot k}{(1+\epsilon)^p}  \cdot \|v\|_2^p \right] \leq  \Pr \left[ \|Gv\|_p^p \leq \frac{c_p\cdot k}{1+p\epsilon/2}  \cdot \|v\|_2^p \right] 
 \leq 
 {\e^{-c'\cdot 2^{-p} \cdot p^{-2} \cdot k\cdot(p\epsilon/(2+p\epsilon))^2}},
\]
where the first inequality derives from a power expansion, $c'>0$ is a universal constant, and we have used $1-\delta=1/(1+p\epsilon/2)$. 

In order to bound the probability of contraction among all distances, we argue that it suffices to use the strong bound on distance contraction, which is derived in Lemma~\ref{LemConP}, and the weak bound on distance expansion from Corollary~\ref{CorWeakUpT} or Lemma~\ref{LemWeakUpT2}, for a $\delta$-dense set $N\subset \SS^{d-1}$, for $\delta$ to be specified later. First, a simple volumetric argument \cite{HIM12} shows that there exists $N\subset \SS^{d-1}$ such that, for all vectors $ x \in \SS^{d-1}$, there exists a vector $ y\in N$ $\|x-y\|_2\leq \delta$, and $|N|= O\left({1}/{\delta}\right)^d$. 

The rest of the proof distinguishes the two cases in the statement.

\paragraph{The first case is $p>2$.}
From now on, suppose that for any $u\in N$, $\|Gu\|_p > {(c_p\cdot k)^{1/p}}/{(1+\epsilon)}$ and $\|Gu\|_p < 2\sqrt{k}$, which are both achieved with probability at least 
\begin{equation}\label{union1}
1- \sum_{u\in N} \left(\Pr\left[\|Gu\|_p\leq {(c_p\cdot k)^{1/p}}/{(1+\epsilon)}  \right] + \Pr\left[ \|Gu\|_p\geq 2\sqrt{k}\right] \right)
\geq 
1- O\left( \frac{1}{\delta}\right)^d\cdot{\e^{-c' \cdot 2^{-p} \cdot p^{-2}\cdot k\cdot(p\epsilon/(2+p\epsilon))^2}}, 
\end{equation}
for a constant $c'>0$, 
by taking a union bound over all vectors in $N$, and employing the aforementioned probability bound for the event 
$\|Gu\|_p\leq {(c_p\cdot k)^{1/p}}/{(1+\epsilon)}$ and Corollary~\ref{CorWeakUpT} for the event 
$\|Gu\|_p\geq 2\sqrt{k}$. 

Let $x$ be an arbitrary vector in $\RR^d$ such that $\|x\|_2=1$. Then, there exists $u\in N$ such that $\|x-u\|_2\leq\delta$. By the triangular inequality we obtain the following:
\begin{equation}\label{ineq1}
\|Gx\|_p\leq \|Gu \|_p+\|G(x-u)\|_p=\|Gu \|_p+\|x-u\|_2 \left\|G\frac{(x-u)}{\|x-u\|_2}\right\|_p\leq 
\|Gu \|_p+\delta \left\|G\frac{(x-u)}{\|x-u\|_2}\right\|_p .
\end{equation}
Let $M=\max_{x\in\SS^{d-1}} \|Gx\|_p$. 
The existence of $M$ is implied by the fact that $\SS^{d-1}$ is compact and $x\mapsto \|x\|_p$, $x\mapsto Gx$ are continuous functions.
Then, by substituting $M$ into Inequality~(\ref{ineq1}), one obtains
\[
M\leq\|Gu\|_p+\delta M \implies M\leq \frac{\|Gu\|_p}{1-\delta}\leq  \frac{2\sqrt{k}}{1-\delta},
\]
where the last inequality holds conditioned on the event analyzed in Equation~(\ref{union1}). 

Let us apply again the triangular inequality, namely
$\|Gu\|_p \leq \|G(u-x)\|_p + \|Gx\|_p =|-1| \cdot\|G(x-u)\|_p +\|Gx\|_p$, which implies   
\[
\|Gx\|_p \geq \|Gu\|_p-\|G(x-u)\|_p \geq \frac{(c_p \cdot k)^{1/p}}{1+\epsilon}-\frac{2 \delta \sqrt{k}}{1-\delta}\geq \frac{1-\epsilon/2}{1+\epsilon} \cdot (c_p \cdot k)^{1/p},
\]
for $\delta \leq \frac{\epsilon\cdot (c_p \cdot k )^{1/p}}{2\sqrt{k}+\epsilon \cdot(c_p\cdot k)^{1/p}}$.
Notice that 
\[
\frac{1}{\delta}=O\left(\frac{k^{1/2-1/p}}{p\epsilon}\right)+1,\]
and by substituting this bound into Equation~(\ref{union1}), we obtain the desired lower bound on the probability of success.  

\paragraph{The second case is $p\in[1,2]$.}
Now, it is possible to apply a better bound on the distance expansion, namely Lemma~\ref{LemWeakUpT2}. Let us assume that, for any $u\in N$, it holds $\|Gu\|_p > {(c_p\cdot k)^{1/p}}/{(1+\epsilon)}$ and $\|Gu\|_p < (3 \cdot c_p \cdot k)^{1/p}$, which are both achieved with probability at least 
\begin{equation}\label{union2}
1-\sum_{u\in N} \left( \Pr\left[ \|Gu\|_p\leq {(c_p\cdot k)^{1/p}}/{(1+\epsilon)}\right] + \Pr \left[ \|Gu\|_p\geq (3 \cdot c_p \cdot k)^{1/p} \right] \right) \geq
1- O\left( \frac{1}{\delta}\right)^d\cdot{\e^{-c' \cdot k\cdot(p\epsilon/(2+p\epsilon))^2}},
\end{equation}
for a constant $c'>0$, by taking again a union bound over all vectors in $N$.  
Once again, we use Inequality~(\ref{ineq1}) to obtain:
\[
M\leq \frac{\|Gu\|_p}{1-\delta} \leq \frac{(3 \cdot c_p \cdot k)^{1/p}}{1-\delta}
\implies 
\|Gx\|_p\geq \|Gu\|_p-\|Gx-Gu\|p\geq (c_p \cdot k)^{1/p} \left( \frac{1}{1+\epsilon} - \frac{3^{1/p}\cdot \delta}{1-\delta}\right)\implies\]
\[\implies
\|Gx\|_p\geq (c_p \cdot k)^{1/p}  \cdot \frac{1-\epsilon/2}{1+\epsilon},
\]
for $\delta = \epsilon/10 \leq  \epsilon/(6(1+\epsilon)+\epsilon)$. 
By setting $1/\delta=O(1/\epsilon)$ into Equation~(\ref{union2}), we obtain the desired bound on the probability of success.
\end{proof}

Theorem~\ref{ThmEmb} implies that the ANN problem for $\ell_p$-products of $\ell_2$ translates to the ANN problem for $\ell_p$-products of $\ell_p$. The latter easily translates to the ANN problem in $\ell_p^{d'}$. 
One can then solve the approximate {\em near} neighbor decision problem in $\ell_p^{d'}$, by approximating $\ell_p^{d'}$ balls of radius~$1$ with a regular grid of side length $\epsilon/(d')^{1/p}$. 
Each approximate ball is essentially a set of $O(1/\epsilon)^{d'}$ cells \cite{HIM12}: each cell either contains an index to their respective approximate near neighbor or it is empty. Now, storing  non-empty cells in a hashtable suffices for queries: each query is either mapped to an existing bucket in the hashtable, which contains an acceptable answer, or it does not belong to any existing bucket, which implies that all data points are approximately far.     
Building polylogarithmically-many approximate near neighbor data structures for various radii leads to an efficient solution for the ANN problem~\cite{HIM12}. 

 \begin{theorem}\label{Tannseqs}
There exists a data structure which solves the ANN problem for point sequences in $\ell_p$-products of $\ell_2$, and satisfies the following bounds on performance:
 \begin{itemize}
 \item If $~p\in [1,2]$, then space usage and preprocessing time is in \[\tilde{O}(d m n )\times \left( \frac{1}{\epsilon} \right)^{O(m \cdot d \cdot \alpha_{p,\epsilon})},\] query time is in $\tilde{O}(dm\log n)$, and $\alpha_{p,\epsilon}=\log(1/\epsilon)\cdot(2+p\epsilon)^2\cdot (p\epsilon)^{-2} $.
 \item 
 If $~2<p<\infty$, then space usage and preprocessing time is in \[\tilde{O}(d m n )\times \left( \frac{d}{p\epsilon} +2\right)^{O(m \cdot d \cdot  \alpha_{p,\epsilon})},\] query time is in $\tilde{O}\left(dm\cdot 2^p\log n\right)$, and $\alpha_{p,\epsilon}=2^p  \cdot\log(1/\epsilon)\cdot(2+p\epsilon)^2\cdot \epsilon^{-2} $.
 \end{itemize}
 We assume $\epsilon \in (0,1/2]$. 
 The probability of success is $\epsilon/2$ and can be amplified to $1-\delta$, by building $\Omega(\log(1/\delta)/\epsilon)$ independent copies of the data-structure. 
\end{theorem}

 \begin{proof} 
For any vector $v$, $(v)_i$ denotes its $i$th element. For any two vectors $v,u$, $v \oplus u$ denotes the concatenation of the two vectors, and for vectors $v_1,\ldots, v_m$, $\bigoplus_i^m v_i$ is equivalent to $v_1 \oplus v_2\oplus \dots
\oplus v_m$. 
Let $G$ be a $k\times d$ matrix with i.i.d.~random variables following $N(0,1)$.
Matrix $G$ provides the random projection from points in $\ell_2$ to points in $\ell_p$. 
 Let $\delta_{p,\epsilon}=p\epsilon/(2+p \epsilon)$. We first consider the case $p>2$. 
 We employ Theorem~\ref{ThmEmb} and we map point sequences in $\ell_2^d$ to point sequences in $\ell_p^{k}$, for 
 \[k =\Theta\left(  \frac{d \cdot 2^p \cdot p^2 \cdot  \log \frac{d}{p\epsilon} }{\delta_{p,\epsilon}^2}\right).\] 
 Then, we treat points as vectors in the projected space and  
 we concatenate them: for each sequence of $m$ points in $\RR^k$, we obtain a vector in $\RR^{km}$. We now argue that building a data structure for the ANN problem in $\ell_p^{km}$ suffices, because the $\ell_p$-product-of-$\ell_p$ distance between two point sequences is equal to the $\ell_p$ distance of the two vectors produced by concatenating the two sequences. To see that, consider the distance between two sequences of $m$ points in $\RR^k$, denoted by $x_1,\ldots,x_m$ and $y_1,\ldots,y_m$:
 \[
\left( \sum_{i=1}^{m} \|x_i-y_i\|_p^p \right)^{1/p}=
\left( \sum_{i=1}^{m} \left(\left(\sum_{j=1}^d |(x_i)_j-(y_i)_j|^p \right)^{1/p}\right)^{p} \right)^{1/p}=
\left( \sum_{i=1}^{m} \sum_{j=1}^d |(x_i)_j-(y_i)_j|^p  \right)^{1/p}=\]
\[=
\left( \left\| \bigoplus_{i=1}^m x_i - \bigoplus_{i=1}^m y_i \right\|_p^p \right)^{1/p}.
 \]

 Now we analyze the probability that no false positives or false negatives occur by the random projection. 
 Fix a query point sequence $Q=q_1,\ldots,q_m\in\left(\RR^{d}\right)^m$ and consider 
 its nearest neighbor $U_*=u_1,\ldots,u_m\in\left(\RR^{d}\right)^m$. 
 By a union bound, the probability of failure for the embedding is at most 
 \[
 \Pr \left[ \exists v \in \RR^d:~\|Gv\|_p \leq \frac{(c_p\cdot k)^{1/p} }{1+\epsilon} \cdot \|v\|_2 \right]+\Pr\left[\sum_{i=1}^{m}\|Gu_i-Gq_i\|_p^p
 \leq (1+\epsilon)^p \cdot c_p \cdot k\sum_{i=1}^m\|u_i-q_i\|_2^p 
 \right].
 \]
By Theorem~\ref{ThmEmb}, the first probability is $\leq\epsilon/10$. Hence, we now bound the second probability. Notice that 
 \[
 \EE \left[\sum_{i=1}^{m}\|Gu_i-Gq_i\|_p^p \right]=
 \sum_{i=1}^{m} \EE \left[\|G(u_i-q_i)\|_p^p \right]= c_p\cdot k \sum_{i=1}^{m} \|u_i-q_i\|_2^p,
 \]
 where the last equality holds by Lemma~\ref{ClmExp}. 
 By Markov's inequality, we obtain,
 \[
 Pr\left[\sum_{i=1}^{m}\|Gu_i-Gq_i\|_p^p
 \leq (1+\epsilon)^p \cdot c_p \cdot k\sum_{i=1}^m\|u_i-q_i\|_2^p 
 \right]\leq (1+\epsilon)^{-p}.
 \] 
 Hence, the total probability of failure is $\frac{1+\epsilon/10}{(1+\epsilon)^p}\leq \frac{1+\epsilon/10}{1+\epsilon}$, and hence the probability of success is at least $1-\frac{1+\epsilon/10}{1+\epsilon}\geq \epsilon/2$.
 In the projected space and after concatenation, in order to solve the $\ell_p^{km}$ instance, we build AVDs \cite{HIM12}. The total space usage, and the preprocessing time is 
 \[\tilde{O}(d m n )\times O(1/\epsilon)^{km}=\tilde{O}(d m n )\times \left( \frac{d}{p\epsilon} +2\right)^{O(m \cdot d \cdot 2^p \cdot p^2 \cdot \log(1/\epsilon)/\delta_{p,\epsilon}^2)}.\] 
  The query time is $O((km)\log n)=\tilde{O}(dm2^p\log n)$. The probability of success can be amplified by repetition. By building $\Theta\left(\frac{\log (1/\delta)}{\epsilon}\right)$ data structures as above, the probability of failure becomes $\delta$.

The same reasoning is valid in the case $p\in[1,2]$, but it suffices to set \[k=\Theta\left(\frac{d\log \frac{1}{\epsilon}}{\delta_{p,\epsilon}^2}\right).\] 
  
 \end{proof}
When $p\in[1,2]$, we can also utilize "high-dimensional" solutions for $\ell_p$ and obtain data structures with complexities polynomial in $d\cdot m$. These data structures are particularly interesting when the complexity of the point sequences $dm$ is considerably higher than $\log n$, since in that case the data structure of Theorem \ref{Tannseqs} requires a prohibitively large amount of storage.   
Combining Theorem~\ref{ThmEmb} with the data structure of \cite{ALRW17}, we obtain the following result.
\begin{theorem}\label{TannHseqs}
There exists a data structure which solves the ANN problem for point sequences in $\ell_p$-products of $\ell_2$, $p\in[1,2]$, and satisfies the following bounds on performance: 
space usage and preprocessing time is in $\tilde{O}(n^{1+\rho_u}+dnm),$ and the query time is in $\tilde{O}(n^{\rho_q} + dm)$, where $ \rho_q,\rho_u$ satisfy:
\[
(1+\epsilon)^p \sqrt{\rho_q}+((1+\epsilon)^p-1)\sqrt{\rho_u}\geq \sqrt{2(1+\epsilon)^p-1}.
\]
 We assume $\epsilon \in (0,1/2]$. 
 The probability of success is $\epsilon/10$ and can be amplified to $1-\delta$, by building $\Omega(\log(1/\delta)/\epsilon)$ independent copies of the data-structure. 
\end{theorem}
\begin{proof}
We proceed as in the proof of Theorem \ref{Tannseqs}. We employ Theorem~\ref{ThmEmb} and by Markov's inequality, we obtain:
 \[
 Pr\left[\sum_{i=1}^{m}\|Gv_i-Gu_i\|_p^p
 \leq (1+\epsilon)^p \cdot c_p \cdot k\sum_{i=1}^m\|v_i-u_i\|_2^p 
 \right]\leq (1+\epsilon)^{-p}.
 \] 
 Then, by concatenating vectors, we map point sequences to points in $\ell_p^{km}$, where $k=\tilde{O}(d)$. For the mapped points in $\ell_p^{km}$, we build the LSH-based data structure from \cite{ALRW17} which succeeds with high probability $1-o(1)$. By independence, both the random projection and the LSH-based structure succeed with probability $(\epsilon/2)\times(1-o(1))\geq \epsilon/10$. Finally we need an additional space of $O(dnm)$ to store and read the input.
\end{proof}

 \section{Polygonal Curves}\label{Scurves}
 
 In this section, we show that one can solve the ANN problem for a certain class of distance functions defined on polygonal curves. Since this class is related to $\ell_p$-products of $\ell_2$, we invoke results of Section~\ref{Sseqs}, and we show an efficient data structure for the case of ``short" curves, i.e.\ when $m$ is relatively small compared to the other complexity parameters. 
 
First, we need to introduce a formal definition of the traversal of two curves.  
 \begin{definition}
   Given polygonal curves $V=v_1, \ldots, v_{m_1}$, $U=u_1, \ldots, u_{m_2}$, a traversal $T= [ (i_1,j_1),\ldots,(i_t,j_t) ] $ is a sequence of pairs of indices referring to a pairing of vertices from the two curves such that:
\begin{enumerate}
 \item $i_1,j_1=1$, $i_t=m_1$, $j_t=m_2$. 
 \item $\forall (i_k, j_k)\in T:$ $i_{k+1}-i_k \in \{0,1\}$ and $j_{k+1}-j_k \in \{0,1\}$.
 \item $\forall (i_k, j_k)\in T:$ $(i_{k+1}-i_k)+(j_{k+1}-j_k)\geq1$.
\end{enumerate} 
 \end{definition} 
 Let us define a class of distance functions for polygonal curves. In this definition, it is implied that we use the Euclidean distance to measure distance between any two points. However, the definition could be easily generalized to arbitrary metrics.

\begin{definition}[$\ell_p$-distance of polygonal curves]\label{Ddist}
 Given polygonal curves $V=v_1, \ldots, v_{m_1}$, $U=u_1, \ldots, u_{m_2}$, we define the $\ell_p$-distance between $V$ and $U$ as the following function:
 \[
 \d_p(V,U)= \min_{T\in\mathcal{T}} \left( \sum_{(i,j)\in T} \| v_{i}-u_{j}\|_2^p  \right)^{1/p},
 \]
 where $\mathcal{T}$ denotes the set of all possible traversals for $V$ and $U$. 
 \end{definition}
The above class of distances for curves includes some widely known distance functions. For instance, $\d_{\infty}(V,U)$ coincides with the DFD of $V$ and $U$ (defined for the Euclidean distance). Moreover $\d_{1}(V,U)$ coincides with DTW for curves $V$, $U$.  

\begin{theorem}\label{Tanncurves}
Suppose that there exists a randomized data structure for the ANN problem in $\ell_p$ products of $\ell_2$, with space in $S(n)$, preprocessing time $T(n)$ and query time $Q(n)$, with probability of failure less than $2^{-4m-1}$. Then, there exists a data structure for the ANN problem for the $\ell_p$-distance of polygonal curves, $1\leq p<\infty$, with space in $(4\e)^{m+1}\cdot S(n)$, preprocessing time $(4\e)^{m+1}\cdot T(n)$ and query time $(4\e)^{m+1}\cdot Q(n)$, where $m$ denotes the maximum length of a polygonal curve, and the probability of failure is less than $1/2$.
\end{theorem} 
 
 \begin{proof}
 We denote by $X$ the input dataset. 
  Given polygonal curves $V=v_1, \ldots, v_{m_1}$, $Q=q_1, \ldots, q_{m_2}$, 
 and traversal $T= [ (1,1),(i_2,j_2),(i_3,j_3),\ldots, ({m_1},{m_2}) ]$, one can define sequences of $l:=|T|$ points  $V_T=v_{i_1},\ldots,v_{i_l}$, $Q_T=q_{i_1}, \ldots, q_{i_l}$, where we allow for consecutive duplicates, such that $\forall k \in \{1,\ldots,l\}$, $v_{i_k}$ is the $k$-th point in $V_T$ 
 and $q_{j_k}$ is the $k$-th point in $Q_T$, 
 if and only if $(i_k,j_k)\in T$. 
 
 One traversal of $V$, $Q$ is uniquely defined by the following parameters: its length 
 , the set of indices $\{k \in \{1,\ldots,l\} \mid i_{k+1} - i_k =0  \text{ and } j_{k+1}- j_k=1\}$ for which only $Q$ is progressing and the set of indices 
 $\{k \in \{1,\ldots,l\} \mid i_{k+1} - i_k =1  \text{ and } j_{k+1}- j_k=1\}$ for which both $Q$ and $V$ are progressing. 
 We build one ANN data structure, for $\ell_p$-products of $\ell_2$, for each possible such set of parameters. 
 Each data structure contains at most $|X|$ point sequences which correspond to curves that are compatible to the corresponding set of parameters. 
 We denote by $m=\max(m_1,m_2)$. The total number of data structures is upper bounded by
\[
\sum_{l=m}^{2m} \sum_{t=0}^{m} {{l}\choose{t}} \cdot {{l-t}\choose{m-t}}\leq \sum_{l=m}^{2m} \sum_{t=0}^{m} {{l}\choose{t}} \cdot {{l}\choose{m-t}}= \sum_{l=m}^{2m} {{2l}\choose{m}}\leq \sum_{l=m}^{4m} {{l}\choose{m}}={{4m+1}\choose{m+1}}\leq (4\e)^{m+1}.
\]
 
 For any query curve $Q$, we create all possible point sequences (all possible $Q_Ts$) and we perform one query per ANN data structure. We report the best answer. The probability that the building of one of the $\leq (4\e)^{m+1}$ 
data structures is not successful is less than $1/2$ due to a union bound.  
 \end{proof}
 
We now investigate applications of the above results to the ANN problem, for certain popular distance functions for curves. 
We remark that in the following complexity bounds, we make use of an overestimation of the bound provided by Theorem \ref{Tanncurves}. In particular we use the term $2^{4m}$ instead of $(4e)^m$, for brevity.

\paragraph*{Discrete Fr\'{e}chet Distance.} DFD is naturally included in the distance class of Definition~\ref{Ddist} for $p=\infty$. However, Theorem~\ref{Tanncurves} is valid only when $p$ is bounded. To overcome this issue, $p$ is set to a suitable large value. 
  
 \begin{lemma}
 \label{Lemtransfrechet}
Let $V=v_1,\ldots,v_{m_1} \in \RR^d$ and $U=u_1,\ldots,u_{m_2} \in \RR^d$ be two polygonal curves. Then for any traversal $T$ of $V$ and $U$:
\[(1+\epsilon)^{-1}\cdot\left(\sum_{(i,j)\in T} \|v_{i} -u_{j}\|^{p}\right)^{1/{p}} \leq
\max_{(i,j)\in T} {\|v_{i}-u_{j}\|}\leq \left(\sum_{(i,j)\in T} \|v_{i} -u_{j}\|^{p}\right)^{1/{p}} 
,\]
for $p = \log\left(|T|\right)/\log(1+\epsilon)$.
\end{lemma}
\begin{proof}
For any $x\in \RR^{|T|}$, H\"older's inequality implies that $\|x\|_\infty \leq \|x\|_p \leq \left(|T|\right)^{1/p} \|x\|_{\infty}$. Hence, for $p\geq\log\left(|T|\right)/\log(1+\epsilon)$, 
\[
(1+\epsilon)^{-1}\cdot\left(\sum_{(i,j)\in T} \|v_{i} -u_{j}\|^{p}\right)^{1/{p}} \leq 
|T|^{-1/p}\cdot\left(\sum_{(i,j)\in T} \|v_{i} -u_{j}\|^{p}\right)^{1/{p}} \leq
\max_{(i,j)\in T} {\|v_{i}-u_{j}\|}\leq \left(\sum_{(i,j)\in T} \|v_{i} -u_{j}\|^{p}\right)^{1/{p}} 
\]
\end{proof}
 
\begin{theorem}\label{Tdfdann}

There exists a data structure for the ANN problem for the DFD of curves, with space and preprocessing time in
\[
\tilde{O}(d m^2 n )\times \left( \frac{d}{\log m} +2\right)^{O(m^{O(1/\epsilon)} \cdot d \cdot \log(1/\epsilon))},
\] 
and query time $\tilde{O}(d m^{O(1/\epsilon)}\cdot 2^{4m} \log n)$, where $m$ denotes the maximum length of a polygonal curve, and $\epsilon\in(0,1/2]$. The data structure succeeds with probability $1/2$, which can be amplified by repetition. 
\end{theorem}
\begin{proof}

We combine Theorem~\ref{Tanncurves} with Theorem~\ref{Tannseqs} for 
\[p=  \frac{\log (2m)}{\log(1+\epsilon)}  \leq \frac{2}{\epsilon}\log (2m), \] 
where the inequality holds because $1+\epsilon \geq \e^{\epsilon/2}$ for $\epsilon \in (0,1/2]$. 
Notice that in order to employ the data structure of Theorem~\ref{Tannseqs} into Theorem~\ref{Tanncurves} we need to amplify the probability of success to $1-2^{-4m-1}$. Hence we need $O(m/\epsilon)$ independent constructions which lead to  a data structure for the ANN problem for $\ell_p$-products of $\ell_p$ with a total of space and preprocessing time in \[\tilde{O}(d m^2 n )\times \left( \frac{d}{p\epsilon} +2 \right)^{O(m \cdot d \cdot \alpha_{p,\epsilon})},\] and {query time} in $\tilde{O}(d m^2  2^p \log n)$, where 
$$
\alpha_{p,\epsilon}=2^p \cdot \log(1/\epsilon)\cdot(2+p\epsilon)^2\cdot \epsilon^{-2} =
2^p \cdot p^2 \cdot  \log(1/\epsilon)\cdot(2+p\epsilon)^2\cdot (p\epsilon)^{-2} =
2^{O(p)} \log(1/\epsilon).
$$
By Theorem~\ref{Tanncurves}, it suffices to build $(4\e)^{m+1}=O(2^{4m})$ such data structures. 

Solving the problem for the $\d_p(\cdot,\cdot)$ distance (instead of the $\d_{\infty}(\cdot,\cdot)$ distance) introduces an approximation factor to the already approximate solution of Theorem~\ref{Tannseqs}. 
By Lemma~\ref{Lemtransfrechet}, 
since $p\geq \log |T|/\log(1+\epsilon)$ for any traversal $T$,  this approximation factor is $(1+\epsilon)$, and hence we get an overall approximation factor of $(1+\epsilon)^2$ which can be reduced to $1+\epsilon$ by rescaling $\epsilon \gets \epsilon/4$.

\end{proof}

\paragraph*{Dynamic Time Warping.} DTW corresponds to the $\ell_1$-distance of polygonal curves in Definition \ref{Ddist}. Let us now combine Theorem \ref{Tanncurves} with each of the Theorems~\ref{Tannseqs} and~\ref{TannHseqs}. 
\begin{theorem}\label{Tdtwann}
There exists a data structure for the ANN problem for DTW of curves, with space and preprocessing time 
\[\tilde{O}(d m^2 n )\times \left(\frac{1}{\epsilon} \right)^{O(m \cdot d \cdot \epsilon^{-2})},\] 
and query time $\tilde{O}(d  \cdot 2^{4m}\log n)$, where $m$ denotes the maximum length of a polygonal curve, and $\epsilon\in(0,1/2]$. The data structure succeeds with probability $1/2$, which can be amplified by repetition. 
\end{theorem}

\begin{proof}
We first amplify the probability of success for the data structure of Theorem \ref{Tannseqs} to $1-2^{-4m-1}$. Hence, the data structure for the ANN problem for $\ell_1$-products of $\ell_1$ needs space and preprocessing time in
\[
\tilde{O}(d m^2 n )\times 2^{O(m \cdot d \cdot \alpha_{p,\epsilon})},
\]
and each {query time} costs $\tilde{O}(d m^2  \log n)$, where $\alpha_{p,\epsilon}=\log(1/\epsilon)\cdot(2+\epsilon)^2\cdot (\epsilon)^{-2} $. 
Theorem~\ref{Tanncurves} concludes the proof. 
\end{proof}
\begin{theorem}\label{Tdtwann2}
There exists a data structure for the ANN problem for DTW of curves, with space and preprocessing time 
 $\tilde{O}(d\cdot 2^{4m}n^{1+\rho_u}),$ and the query time is in $\tilde{O}( d\cdot 2^{4m} n^{\rho_q})$, where $ \rho_q,\rho_u$ satisfy:
\[
(1+\epsilon) \sqrt{\rho_q}+\epsilon\sqrt{\rho_u}\geq \sqrt{1+2\epsilon}.
\]
 We assume $\epsilon \in (0,1/2]$. The data structure succeeds with probability $1/2$, which can be amplified by repetition. 
\end{theorem}
\begin{proof}
First, amplify the probability of success for the data structure of Theorem \ref{TannHseqs} to $1-2^{-4m-1}$, by building independently $\tilde{O}(m)$ such data structures. 
We substitute the resulting data structure into Theorem~\ref{Tanncurves}. The resulting space usage and preprocessing time is in 
\[
\tilde{O}\left((4\e)^{m+1} \cdot m(n^{1+\rho_u}+dnm)\right)=\tilde{O}\left(d\cdot 2^{4m} n^{1+\rho_u} \right),
\]
and the query time is in 
\[
\tilde{O}\left((4\e)^{m+1} \cdot m(n^{\rho_q}+dm) \right)=
\tilde{O}\left(d\cdot 2^{4m} n^{\rho_q}  \right).
\]
\end{proof}

\section{Discrete Fr\'{e}chet distance in high dimensions}\label{Sechd}

In this section, we focus on the Discrete Fr\'{e}chet distance, when the dimension of the ambient space is high, i.e.\ $d=\omega(\log n)$. We combine known results on random projections in order to establish a scheme which is sensitive to the intrinsic dimensionality of the vertex-points, and an improved result for the approximate near neighbor problem.

\subsection{Low doubling dimension}
\label{ssection:doub}
For any point $x\in\RR^d$ and a set $X\subseteq \RR^d$, we define $d(x,X)= \inf_{y\in X} \|x-y\|_2$. We also need the following definition of the doubling constant of some arbitrary metric space. 

\begin{definition}
Consider any metric space with ground set $X$ and let $B(p,r)$ be the metric ball centered at $p\in X$ with radius $r$. 
The {\em doubling constant} of $X$, denoted by $\lambda_X$, is the smallest integer $\lambda_X$ such that for any
$p \in X$ and $r > 0$, the ball $B(p, r)$ (in $X$) can be covered by at most $\lambda_X$ balls of radius $r/2$ centered at points in $X$. 
\end{definition}

The notion of doubling dimension of a metric space is also relevant. The \emph{doubling dimension} is equal to the logarithm of the doubling constant. 

\begin{theorem}{\em\cite[Thm~4.1]{IN07}} \label{Thmrandprojdoub}
Let $G$ be a $k\times d$ matrix with i.i.d.~random variables following $N(0,1)$ and let matrix $A=\frac{1}{\sqrt{k}}G$. 
For $X\subseteq \RR^d$, $\epsilon \in (0, 1)$ and $\delta \in (0, 1/2)$, there exists $k = O\left( \frac{\log(2/\epsilon)}{\epsilon^2} \cdot \log (1/\delta) \cdot \log \lambda_X \right)$
such that, for every $x_0 \in X$, with probability at least $1 - \delta$,
\begin{itemize}
    \item[(1)] $d(Ax_0 , A(X \setminus \{x_0 \})) \leq  (1 + \epsilon)\, d(x_0 , X \setminus \{x_0 \})$, 
    \item[(2)] every $x \in X$ with $\|x_0 - x \|_2 > (1 + 2\epsilon)\, d(x_0 , X \setminus \{x_0 \})$ satisfies
\[
\|Ax_0 - Ax\|_2 > (1 + \epsilon)\, d(x_0 , X \setminus \{x_0\}),
\]
where for any set $X$, we let $A(X)=\{Ax \mid x\in X\}$. 
\end{itemize}
\end{theorem}
 \begin{theorem}\label{TannseqsDD}
 
 Let $\lambda_X$ be the doubling constant of the input dataset, i.e.~the doubling constant of the set of all points appearing in the data point sequences. Let $m$ be the (maximum) length of the point sequences. 
There exists a data structure which solves the ANN problem for point sequences in $\ell_{\infty}$-products of $\ell_2$, and satisfies the following bounds on performance: 
space usage and preprocessing time in 
\[
\tilde{O}\left(d m n  \right)\times \left( \frac{\log \lambda_X \cdot \log (2/\epsilon)}{\epsilon^2} +2\right)^{O(m^{O(1/\epsilon)}  \cdot  \log^2(2/\epsilon) \cdot \log \lambda_X)},
\]
query time in $\tilde{O}\left(dm^{O(1/\epsilon)}\cdot \log \lambda_X \cdot \log n\right)$, where $\epsilon \in (0,1/2]$. For any query point sequence, the preprocessing algorithm succeeds with constant probability.
\end{theorem}
\begin{proof}

First, we employ Theorem~\ref{Thmrandprojdoub} for $\delta={1/(2m)}$. Then, by a union bound over the $m$ points in the query point sequence, we have that with probability $1/2$, the approximate nearest neighbor under the $\ell_{\infty}$-product of $\ell_2$ metric, is approximately preserved (as guaranteed by Theorem \ref{TannseqsDD}). Now, we are able to invoke Theorem \ref{Tannseqs} for large enough $p=  \frac{\log m}{1+\epsilon}  \leq 2\epsilon^{-1}\log m$, 
and 
$d= O\left( \frac{\log(2/\epsilon)}{\epsilon^2} \cdot \log m \cdot \log \lambda_X \right)$. This large value of $p$ guarantees that the multiplicative approximation factor which is introduced when solving an $\ell_p$-product instance instead of an $\ell_{\infty}$-product instance is $1+\epsilon$, which is a consequence of H\"older's inequality (as in the proof of Lemma \ref{Lemtransfrechet}). Since we invoke a $1+\epsilon$ approximate solution for the $\ell_p$-product of $\ell_2$ instance, we have an overall approximation of $(1+\epsilon)^2$ which can be reduced to $1+\epsilon$ by rescaling $\epsilon \gets \epsilon/4$.  
\end{proof}

Theorem~\ref{TannseqsDD} has immediate implications for the ANN problem under the DFD. 

\begin{theorem} 
\label{TannDFDdd}
Let $\lambda_X$ be the doubling constant of the input dataset, i.e.~the doubling constant of the set of all points-vertices of the polygonal curves. 
There exists a data structure for the ANN problem for the DFD of curves, with space and preprocessing time in
\[
\tilde{O}\left( d m n  \right)\times \left( \frac{\log \lambda_X \cdot \log (2/\epsilon)}{\epsilon^2} +2\right)^{O(m^{O(1/\epsilon)}  \cdot  \log^2(2/\epsilon) \cdot \log \lambda_X)},
\] 
and query time in $\tilde{O}\left(dm^{O(1/\epsilon)}2^{4m}\cdot \log \lambda_X \cdot \log n\right)$, where $m$ denotes the maximum length of a polygonal curve, and $\epsilon\in(0,1/2]$. The data structure succeeds with probability $1/2$, which can be amplified by repetition. 
\end{theorem}
\begin{proof}
We combine Theorem~\ref{Tanncurves} with Theorem~\ref{TannseqsDD}. 
\end{proof}

\subsection{Approximate near neighbors in high dimensions}
All results so far concern the ANN problem. In this subsection, we focus on the simplified task of deciding the approximate {near} neighbor problem with witness. Given a dataset $\mathcal{P}$ of curves, a radius parameter $r>0$ and an error parameter $\epsilon>0$, the goal is to build a data structure 
which supports the following type of query. For a query curve $Q$:
\begin{itemize}
    \item if there exists a curve in $\mathcal{P}$ at distance $\leq r$ from $Q$, then return a curve in $\mathcal{P}$ at distance $\leq (1+\epsilon)r$,
    \item if all curves in $\mathcal{P}$ are at distance $>(1+\epsilon)r$ from $Q$, then return "no".
\end{itemize}
In short, the data structure returns either a point at distance $\leq (1+\epsilon)r$ from $Q$, or "no". 
In the intermediate case that the nearest curve lies at some distance in $(r, (1+\epsilon)r]$, the data structure returns any answer.

Given a set of $n$ points, if the dimension is $\omega(\log n) $, then one can apply the Johnson-Lindestrauss lemma to reduce the dimension to $O\left(\frac{\log n}{\epsilon^2}\right)$, while probabilistically preserving the Euclidean norms up to $(1\pm \epsilon)$ factors. The main observation in this subsection is that, since we focus on the approximate near neighbor problem, we may use as target space the Hamming space of dimension $O\left(\frac{\log n}{\epsilon^2}\right)$ and, thus, significantly simplify the task of searching in the projection space.  
The reason for which it is possible to employ a Hamming target space is that we only need to ensure the following condition: points with distance below a given threshold must remain near after the projection, while points with distance beyond the threshold must remain far apart.

The main ingredient is a randomized mapping from $\ell_1^d$ to $\{0,1\}^{{O\left(\frac{\log n}{\epsilon^2}\right)}}$. This mapping also appears in~\cite{ACW16}, and it 
resembles ideas which appear in~\cite{KOR00}, and~\cite{AEPS16}. We include the proof for completeness. 

\begin{lemma}\label{LemTohamm}
Let $X\subset \RR^d$ be a set of $n$ points. 
There exists a distribution over mappings  $f:~\RR^{d} \mapsto \{0,1 \}^{O\left(\epsilon^{-2}\log n\right)}$, such that for any $p,q\in X$,
\[
\|p-q\|_1 \leq r \implies \|f(p) -f(q)\|_1 \leq r',
\]
\[
\|p-q\|_1 \geq (1+\epsilon)r \implies \|f(p) -f(q)\|_1 >r',
\]
where $r'$ is a constant depending on the target dimension. The randomized embedding succeeds with high probability. 
\end{lemma}
\begin{proof}
We denote by $F$ the Locality Sensitive Hashing family of~\cite{AI06}, which is 
$(1-\frac{1}{\alpha},1-\frac{c}{c+\alpha},1,c)$-sensitive: 
if $\|x-y\|_1 \leq 1$ then $\Pr_{h\in F}[h(x)=h(y)]\geq 1-\frac{1}{\alpha}$, and if $\|x-y\|_1 \geq c$ then $\Pr_{h\in F}[h(x)=h(y)]\leq 1-\frac{c}{\alpha+c}$. We build the amplified family of functions $G_k=\{g(x)=(h_1(x),\ldots,h_k(x)) : i=1,\ldots, k,~ h_i\in F) \}$. 
Setting $\alpha = k = \log{n}$, we have:
\[
{p_1}={\Big(1-\frac{1}{\alpha}\Big)}^k={\Big(1-\frac{1}{\log{n}}\Big)}^{\log{n}} \geq
 {\Big(\exp \Big(-\frac{1}{\log n-1}\Big)\Big)}^{\log{n}}\geq \frac{1}{\e^{1+o(1)}},
\] 
\[
{p_2}={\Big(1- \frac{c}{\alpha+{c}}\Big)}^k = {\Big(1- \frac{c}{\log n+c}\Big)}^{\log n} .
\]
Hence, 
\[
{p_2} \geq \exp(-c) \geq \frac{1}{\e \cdot (2c-1)},
\]
and 
\[
{p_2} \leq \exp \Big(-\frac{c}{1+\frac{c}{\log n}}\Big)=  \exp \Big(-\frac{c}{1+o(1)}\Big) 
\leq \exp \big(-c+o(1)\big) \leq \frac{\e^{o(1)}}{\e c}.
\]

We first sample $g_1\in G_k$. 
We denote by $g_1(P)$ the image of $P$ under $g_1$, which is a set of nonempty buckets.
Now each nonempty bucket $x\in g_1(P)$ is mapped to $\{ 0,1 \}$: with probability $1/2$,  set $f_1(x)=0$, otherwise set $f_1(x)=1$. 

This is repeated $d'$ times, and eventually for $p\in \RR^d$, we compute the function \[f(p)=(f_1(g_1(p)),\ldots,f_{d'}(g_{d'}(p))),\] where $f: P \to \{0, 1 \}^{d'}$.
Now, observe that 
\[
\|p-q\|_2\leq r \implies \EE  [\|f_i(g_i(p))-f_i(h_i(q))\|_1]  \leq 0.5 (1-p_1), \; i=1,\dots,d'\implies   \EE  [\|f(p)-f(q)\|_1]\leq 0.5\cdot d' \cdot (1-p_1), \]
\[\|p-q\|_2\geq cr \implies \EE  [\|f_i(g_i(p)-f_i(g_i(q))\|_1]\geq 0.5 (1-p_2),\, i=1,\dots,d'\implies  \EE  [\|f(p)-f(q)\|_1]\geq 0.5\cdot d'\cdot (1-p_2). \]

Finally, the lemma holds by standard Chernoff bounds, and by making use of the bounds on $p_1$, $p_2$. 
\end{proof}

\begin{theorem}\label{ThmseqDFDhd}
There exists a data structure which solves the approximate near neighbor problem for point sequences in $\ell_{\infty}$-products of $\ell_2$, and satisfies the following bounds on performance: 
space usage and preprocessing time in 
$O(dnm)+(nm)^{O(m\epsilon^{-2})},$ query time in $\tilde{O}(m \log n)$. For any query point sequence, the preprocessing algorithm succeeds with constant probability.
\end{theorem}
\begin{proof}

First, we randomly project points from $\ell_2^d$ to $\ell_1^{\tilde{O(d)}}$, while approximately preserving all distances up to factors $1\pm \epsilon$ (see e.g.~\cite{Mat08}). Then, we employ Lemma \ref{LemTohamm} to sample a mapping $f(\cdot)$ which allows us to map points to $\{0,1\}^{k}$, $k=O(\epsilon^{-2} \log (nm))$. 
Data point sequences are stored in a hash table (assuming perfect hashing). Any point sequence $x_1,\ldots,x_m$ is associated to a tuple of strings $(f(x_1),\ldots,f(x_m))$ (or equivalently a string of length $km$), which serves as a key.    Now, for each data point sequence $p_1,\ldots,p_m$, we store pointers to all buckets with keys $(t_1,\ldots,t_m)\in \left(\{0,1\}^k\right)^m$, such that $\max_i\|f(p_i)-t_i\|_1 \leq r'$, where $r'$ is defined in Lemma \ref{LemTohamm}. For a query sequence $q_1,\ldots,q_m$, we compute $(f(q_1),\ldots,f(q_m))$ and we probe the hashtable to return a pointer to a near neighbor (if any). 
Hence, we need a total of $O(dmn+n 2^{km})=O(dmn)+(nm)^{O(m\epsilon^{-2})}$ of storage. The query time is $O(dm+km)=\tilde{O}(dm \log n)$. 
\end{proof}

\begin{theorem}\label{TannDFDhd}
There exists a data structure for the approximate near neighbor problem under the DFD of curves, with space and preprocessing time in
$O(dn)+(nm)^{O(m\epsilon^{-2})},$ and query time in $\tilde{O}\left(d\cdot 2^{4m}\cdot \log n\right)$, where $m$ denotes the maximum length of a polygonal curve, and $\epsilon\in(0,1/2]$. The data structure succeeds with probability $1/2$, which can be amplified by repetition. 
\end{theorem}
\begin{proof}
We use the same construction as in the proof of  Theorem~\ref{Tanncurves}. Instead of using an ANN data structure as our main building block, we use the data structure of  Theorem~\ref{ThmseqDFDhd}. We build $O(2^{4m})$ such data structures, one for each traversal, and we probe each one of them for an approximate near neighbor in the $\ell_{\infty}$-products of $\ell_2$ metric. Notice that we only need to store our original dataset once and refer to it with pointers.   
If one of the data structures returns a data curve, then we stop searching and report that curve. If none of them returns a data curve, then we report "no". 
\end{proof}

\section{Conclusion}

Thanks to the simplicity of the approach, we expect it to lead to software of practical interest so we are working on a C++ implementation. We may apply it to real scenarios with data from road segments or time series. 

The key ingredient of our approach is a randomized embedding from $\ell_2$ to $\ell_p$ which is the first step to the ANN solution for $\ell_p$-products of $\ell_2$. The embedding is essentially a Gaussian projection and it exploits the $2$-stability property of normal variables, along with standard properties of their tails. We expect that a similar result can be achieved for $\ell_p$-products of $\ell_q$, where $q\in[1,2)$. One related result for ANN \cite{BG16}, offers a dimension reduction for $\ell_q$,  $q\in[1,2)$. 

\bibliography{jlann}
\end{document}